\documentclass[11pt]{article}
\usepackage{amssymb}
\usepackage{amsthm}
\usepackage{boxedminipage}
\usepackage{ifthen}
\usepackage{amsmath}
\usepackage{authblk}

\usepackage{amssymb,amsfonts,amsthm,epsfig}
\newtheorem{theorem}{Theorem}

\newtheorem{lemma}[theorem]{Lemma}

\newtheorem{corollary}[theorem]{Corollary}

\newtheorem{definition}[theorem]{Definition}

\newcommand{\pr}{{\bf Pr}}
\newcommand{\ex}{{\bf E}}
\newcommand{\etal}{{\it et al.}}
\newcommand{\eps}{{\epsilon}}

\title{A tight lower bound instance for k-means++ in constant dimension}

\author[1]{Anup Bhattacharya}
\author[1]{Ragesh Jaiswal \thanks{Corresponding author: \texttt{rjaiswal@cse.iitd.ac.in}}}
\author[2]{Nir Ailon \thanks{Nir Ailon acknowledges the support of a Marie Curie International Reintegration Grant PIRG07-GA-2010-268403, as well as the support of The Israel Science Foundation (ISF) no. 1271/13.}}

\affil[1]{IIT Delhi, India}
\affil[2]{Technion, Haifa, Israel}

\date{}

\begin{document}
\maketitle
\begin{abstract}
The k-means++ seeding algorithm is one of the most popular algorithms that is used for finding the initial $k$ centers when using the k-means heuristic. The algorithm is a simple sampling procedure and can be described as follows: 
\begin{quote}
Pick the first center randomly from the given points. 
For $i > 1$, pick a point to be the $i^{th}$ center with probability proportional to the square of the Euclidean distance of this point to the closest previously $(i-1)$ chosen centers.
\end{quote}
The k-means++ seeding algorithm is not only simple and fast but also gives an $O(\log{k})$ approximation in expectation as shown by Arthur and Vassilvitskii~\cite{ArthurV07}.
There are datasets~\cite{ArthurV07,AggarwalDK09} on which this seeding algorithm gives an approximation factor of $\Omega(\log{k})$ in expectation. 
However, it is not clear from these results if the algorithm achieves good approximation factor with reasonably high probability (say $1/poly(k)$). 
Brunsch and R\"{o}glin~\cite{br12} gave a dataset where the k-means++ seeding algorithm achieves an  $O(\log{k})$ approximation ratio with probability that is exponentially small in $k$. 
However, this and all other known {\em lower-bound examples} \cite{ArthurV07,AggarwalDK09} are high dimensional. 
So, an open problem was to understand the behavior of the algorithm on low dimensional datasets. 
In this work, we give a simple two dimensional dataset on which the seeding algorithm achieves an $O(\log{k})$ approximation ratio with probability exponentially small in $k$. 
This solves open problems posed by Mahajan \etal~\cite{mnv12} and by Brunsch and R\"{o}glin~\cite{br12}.
\end{abstract}


\section{Introduction}

The k-means clustering problem is one of the most important problems in Data Mining and Machine Learning that has been widely studied. The problem is defined as follows:
\begin{quote}
{\bf (k-means problem)}: Given a set of $n$ points $X = \{x_1, ..., x_n\}$ in a $d$-dimensional space, find a set of $k$ points $C = \{c_1, ..., c_k\}$ (these are called {\em centers}) such that the cost function $\Phi_{C}(X) = \sum_{x \in X} \min_{c \in C} D(x, c)$ is minimized. 
Here $D(x,c)$ denotes the square of the Euclidean distance between points $x$ and $c$.
In the {\em discrete} version of this problem, the centers are constrained to be a subset of the given points $X$.
\end{quote}

The problem is known to be NP-hard even for small values of the parameters such as when $k=2$~\cite{d07} and when $d=2$~\cite{v09,mnv12}.
There are various approximation algorithms for the problem.
However, in practice, a heuristic known as the k-means algorithm (also known as Lloyd's algorithm) is used because of its excellent performance on real datasets even though it does not give any performance guarantees. 
This algorithm is simple and can be described as follows: 
\begin{quote}
{\bf (k-means Algorithm)}: (i) Arbitrarily, pick $k$ points $C$ as centers. (ii) Cluster the given points based on the nearest distance to centers in $C$. (iii) For all clusters, find the mean of all points within a cluster and replace the corresponding member of $C$ with this mean. Repeat steps (ii) and (iii) until convergence.
\end{quote}

Even though the above algorithm performs very well on real datasets, it guarantees only convergence to local minima.
This means that this {\em local search} algorithm may either converge to a local optimum solution or may take a large amount of time to converge~\cite{ArthurV06,ArthurV06b}. 
Poor choice of the initial $k$ centers (step (i)) is one of the main reasons for its bad performance with respect to approximation factor. 
A number of {\em seeding} heuristics have been suggested for choosing the initial centers. 
One such seeding algorithm that has become popular is the k-means++ seeding algorithm.
The algorithm is extremely simple and runs very fast in practice.
Moreover, this simple randomized algorithm also gives an approximation factor of $O(\log{k})$ in expectation~\cite{ArthurV07}.
In practice, this seeding technique is used for finding the initial $k$ centers to be used with the k-means algorithm and this ensures a theoretical approximation guarantee.
The simplicity of the algorithm can be seen by its simple description below:
\begin{quote}
{\bf (k-means++ seeding)}: Pick the first center randomly from the given points. 
After picking $(i-1)$ centers, pick the $i^{th}$ center to be a point $p$ with probability proportional to the square of the Euclidean distance of $p$ to the closest previously $(i-1)$ chosen centers.
\end{quote}

A lot of recent work has been done in understanding the power of this simple sampling based approach for clustering.
We discuss these in the following paragraph.

\subsection{Related work} 
Arthur and Vassilvitskii~\cite{ArthurV07} showed that the sampling algorithm gives an approximation guarantee of $O(\log{k})$ in expectation. They also give an example dataset on which this approximation guarantee is best possible. Ailon \etal~\cite{AJMonteleoni09} and Aggarwal \etal~\cite{AggarwalDK09} showed that sampling more than $k$ centers in the manner described above gives a constant {\em pseudo-approximation}.\footnote{Here pseudo-approximation means that the algorithm is allowed to output more than $k$ centers but the approximation factor is computed by comparing with the optimal solution with $k$ centers.}
Ackermann and Bl\"{o}mer~\cite{ab10} showed that the results of Arthur and Vassilvitskii~\cite{ArthurV07} may be extended to a large class of other distance measures. 
Jaiswal \etal~\cite{jks12} showed that the seeding algorithm may be appropriately modified to give a $(1 + \epsilon)$-approximation algorithm for the k-means problem. 
Jaiswal and Garg~\cite{jg12} and Agarwal \etal~\cite{ajp13} showed that if the dataset satisfies certain separation conditions, then the seeding algorithm gives constant approximation with probability $\Omega(1/k)$. 
Bahmani \etal~\cite{b12} showed that the seeding algorithm performs well even when fewer than $k$ sampling iterations are executed provided that more than one center is chosen in a sampling iteration.
We now discuss our main results.

\subsection{Main results}
The lower-bound examples of Arthur and Vassilvitskii~\cite{ArthurV07} and Aggarwal \etal~\cite{AggarwalDK09} have the following two properties: (a) the examples are high dimensional and (b) the examples lower-bound the {\em expected} approximation factor. 
Brunsch and R\"{o}glin~\cite{br12} showed that the k-means++ seeding gives  
an approximation ratio of at most $(2/3 - \epsilon)\cdot \log{k}$ only with probability that is exponentially small in $k$.
They constructed a high dimensional example where this is not true and showed that an $O(\log{k})$ approximation is achieved with probability exponentially small in $k$.
An important open problem mentioned in their work is to understand the behavior of the seeding algorithm on low-dimensional datasets. 
This problem is also mentioned as an open problem by Mahajan \etal~\cite{mnv12} who showed that the {\em planar} (dimension=2) k-means problem is NP-hard.
In this work, we construct a two dimensional dataset on which the k-means++ seeding algorithm achieves an approximation ratio $O(\log{k})$ with probability exponentially small in $k$. 
More formally, here is the main theorem that we prove in this work.

\begin{theorem}[Main Theorem]\label{thm:main}
Let $r(k) = \delta \cdot \log{k}$ for a fixed real $\delta \in (0,\frac{1}{120})$. 
There exists a family of instances for which k-means++ achieves an $r(k)$-approximation with probability at most $2^{-k } + e^{\left(-(k-1)^{1- 120 \delta- o(1)}\right)}$.
\end{theorem}

Note that the theorem refutes the conjecture by Brunsch and R\"{o}glin~\cite{br12}.
They conjectured that the k-means++ seeding algorithm gives an $O(\log{d})$-approximation for any $d$-dimensional instance.

\subsection{Our techniques}
All the known lower-bound examples~\cite{ArthurV07,AggarwalDK09,br12} have the following general properties:
\begin{enumerate}
\item[(a)] All optimal clusters have equal number of points.
\item[(b)] The optimal clusters are high dimensional simplices.
\end{enumerate}
In order to construct a counterexample for the two dimensional case, we consider datasets that have different number of points in different optimal clusters. 
Our counterexample is shown in Figure~\ref{fig1}.
The optimal clusters (indicated in the figure using shaded areas) are along the vertical lines drawn along the $x$-axis.
In the next section, we will show that these are indeed the optimal clusters.
Note that the cluster sizes decrease exponentially going from left to right.
We say that an optimal cluster is {\em covered} by the algorithm if the algorithm picks a center from that optimal cluster.
We will use the following two high level observations to show the main theorem:
\begin{itemize}
\item {\bf Observation 1}: The algorithm needs to cover more than a certain minimum fraction of clusters to achieve a required approximation.

\item {\bf Observation 2}: After any number of iterations, the probability of sampling the next center from an uncovered cluster is not too large compared to the probability of sampling from a covered cluster.
\end{itemize}

We bound the probability of covering more than a certain minimum fraction of clusters by analyzing a simple Markov chain. This Markov chain is almost the same as the chain used by Brunsch and R\"{o}glin~\cite{br12}. 
We also borrow the analysis of the Markov chain from \cite{br12}.
So, in some sense, the main contribution of this paper is to come up with a two dimensional instance the analysis of which may be reduced to the Markov chain analysis in \cite{br12}.

In the next section, we give the details of our construction and proof.

\section{The Bad Instance} 

We provide a family of $2$-dimensional instances on which performance of k-means++ is bad in the sense of Theorem~\ref{thm:main}. 
This family is depicted in Figure \ref{fig1}. 
We first recursively define certain quantities that will be useful in describing the construction. 
Here $m$ is any positive integer, $r$ is any positive real number, and $\Delta$ is a positive real number dependent on $k$ (we will define this dependency later during analysis).
\begin{eqnarray*}
&& r_1 = r \quad \textrm{and} \quad \forall i, 2 \leq i < k, r_i = 2 \cdot r_{i-1} \\
&& m_1 = m \quad \textrm{and} \quad \forall i, 2 \leq i < k, m_i =  (1/4) \cdot m_{i-1}
\end{eqnarray*}

Note that the input points may overlap in our construction.
We will consider $k$ {\em groups} of points $G_0, ..., G_{k-1}$.
These groups are shown as shaded areas in Figure~\ref{fig1}.
They are located at only $k$ distinct $x$-coordinates.
These $k$ distinct $x$-coordinates are given by $(x_0, x_1, ..., x_{k-1})$, where 
$x_0 = 0, x_1 = \Delta \cdot r_1, x_2 = \Delta \cdot (r_1 + r_2), ..., x_{k-1} = \Delta \cdot (r_1 + ... + r_{k-1})$.
The $i^{th}$ group, $G_i$, consists of points that have the $x$-coordinate $x_i$.
We will later show that $G_0, ..., G_{k-1}$ is actually the optimal k-means clustering for our instance.
Group $G_0$ has $12 k 2^k m$ points located at $(x_0, 0)$. 
For all $i \geq 1$, group $G_i$ has $4km_i$ points located at $(x_i, 0)$, and for all $0 \leq j < k$, $G_i$ has $\frac{m_i}{4^j}$ points located at each of $\left(x_i, 2^j r_i \right)$ and $\left(x_i, - 2^j r_i \right)$.

Let the total number of points on $i^{th}$ group be denoted by $M_i$. 
Therefore, we can write summing points across all locations on that cluster to get the following:
\begin{eqnarray}\label{eqn:mass}
\forall i \geq 1, M_i &=& 4 k m_i  + 2m_i + 2 (m_i/4)  + ... + 2 (m_i/4^{k-1}) \nonumber \\
&=& 4 k m_i + 2 m_i (1 + 1/4 + ... + 1/4^{k-1})
\end{eqnarray}

Note that $M_{i+1}=M_i/4$.
\begin{figure}[ht]        
        \label{fig1}
        \begin{center}
                \includegraphics[scale=0.4]{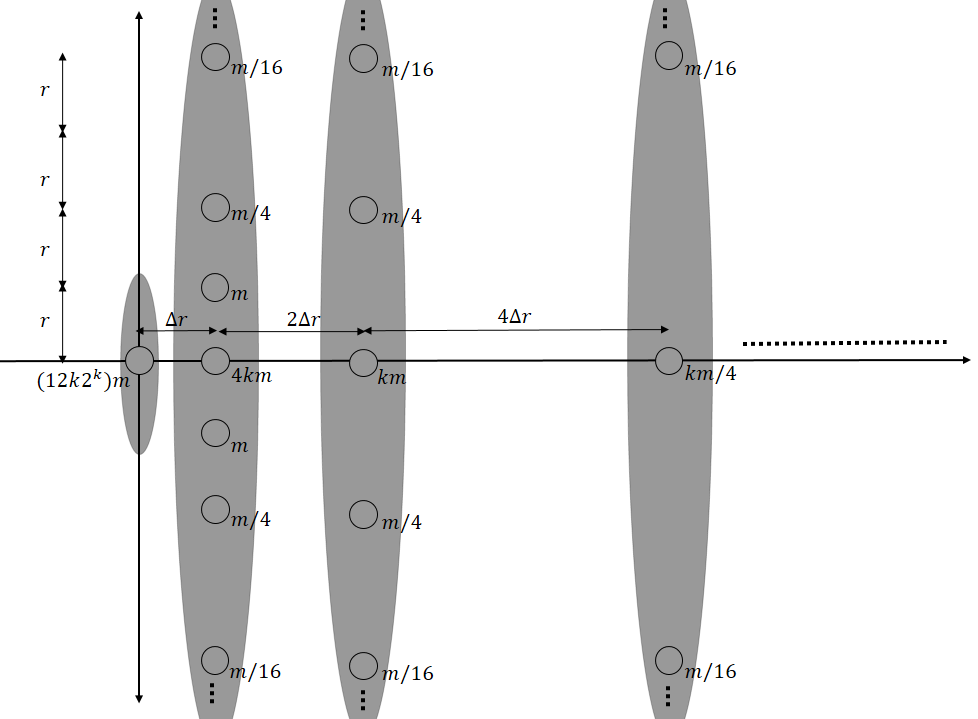}
        \end{center}
        \caption{$2$-D example instance showing the $0^{th}$, $1^{st}$, $2^{nd}$, and $3^{rd}$ optimal clusters only. Note that this figure is not to scale.}
\end{figure}


\subsection{Optimal solution for our instance} 
We consider the following partitioning of the given points: 
Let $H_0$ denote the subset of points on the $x$-axis and for $|i| \geq 1$,
let $H_i$ denote the subset of all points that are located at $y$-coordinate $sgn(i) \cdot 2^{|i|-1} \cdot r$.
For any point $p \in H_i$, we say that point $p$ is in {\em level} $i$.
Given the above definitions of group and level, the location of a point may be defined by a tuple $(i, j)$, where $i$ denotes the index of the group to which this point belongs and $j$ denotes the level of this point.

Given a set $C$ of centers and a subset of points $Y$, the potential of $Y$ with respect to $C$ is given by  $\Phi_{C}(Y) = \sum_{y \in Y} \min_{c \in C} D(y, c)$. 
Furthermore, the potential of a location $l = (i, r)$ with respect to $C$ is defined as $\Phi_{C}(l) = \sum_{p \textrm{ located at } l} D(p, C)$.
Here, $D(p, C) = \min_{c \in C} D(p, c)$.
Given a set of locations $L = \{l_1, ..., l_s\}$ and a subset of points $Y$, $\Phi_{L}(Y)$ denotes the potential of points $Y$ with respect to a set of centers located at locations in $L$.

We start by showing certain basic properties of our instance.

\begin{lemma}\label{lemma:2}
Let $|j| \geq i > 0$. The total number of points  at level $j$ in group $i$ is $\frac{m}{4^{|j| - 1}}$.
\end{lemma}
\begin{proof}
The points for group $i$ are located at distance $\pm r_i, \pm 2r_i, \pm 4r_i, \ldots$.
Since $r_i = 2^{i-1} \cdot r$, this means that the points in $G_i$ are located at $\pm 2^{i-1} r, \pm 2^{i} r, \pm 2^{i+1} r, \ldots$.
So, the number of points is given by $\frac{m_i}{4^{|j|-i}} = \frac{m}{4^{|j|-1}}$.
\end{proof}

\begin{lemma}\label{lemma:3}
For all $i > 0$ and $|j| > 0$, 
\begin{equation*}
	\Phi_{\{(i, 0)\}}(H_{j})  - \Phi_{\{(i, j)\}}(H_{j}) \leq kmr^2
\end{equation*}
\end{lemma}
\begin{proof}
From Lemma~\ref{lemma:2}, we know that the total number of points at level $j$ of any group is either $0$ or $m/4^{|j|-1}$. 
The net change in the squared Euclidean distance of any point in $H_j$ with respect to locations $(i, 0)$ and $(i, j)$ is $(2^{|j|-1} r)^2$. 
So, the total change in potential is at most $k \cdot \frac{m}{4^{|j|-1}} \cdot (2^{|j|-1} r)^2 = k m r^2$.
\end{proof}

\begin{lemma}\label{lemma:4}
For all $i > 0$ and $|j| > 0$,  
\begin{eqnarray*}
&&\Phi_{\{(i, 0)\}}(H_{j+1\cdot sgn(j)} \cup H_{j+2\cdot sgn(j)} \cup ...) \leq\\
&&~~~\Phi_{\{(i, j)\}}(H_{j+1\cdot sgn(j)} \cup H_{j+2\cdot sgn(j)} \cup ...)+2kmr^2.
\end{eqnarray*}
\end{lemma}
\begin{proof}
WLOG assume that $j > 0$. 
From Lemma~\ref{lemma:2}, we can get an upper bound in the following manner:
\begin{eqnarray*}
&&\!\!\!\!\!\!\!\!\!\!\Phi_{\{(i, 0)\}}(H_{j+1} \cup H_{j+2} \cup \cdots) - \Phi_{\{(i, j)\}}(H_{j+1} \cup H_{j+2} \cup \cdots) \\
&\leq& \sum_{t=1}^{k} \sum_{l=j+1}^{\infty} \frac{m}{4^{l-1}} \cdot r^2 \left((2^{l-1})^2 - (2^{l-1} - 2^{j-1})^2\right) \\
&=& k m r^2 \cdot  \sum_{l=j+1}^{\infty} \frac{1}{4^{l-1}} \cdot \left(2^{l+j-1} - 2^{2j-2}\right) \\
&\leq& k m r^2 \cdot  \sum_{l=j+1}^{\infty} 2^{j - l + 1} \\
&\leq& 2 k m r^2
\end{eqnarray*}
\end{proof}

Let $C$ denote a set of optimal centers for the k-means problem. 
Let $L$ denote the set of locations of these centers.
We will show that $L = \{(0, 0), (1, 0), ..., (k-1, 0)\}$.
We start by showing some simple properties of the set $L$.
We will need the following additional definitions:
We say that a group $G_i$ is {\em covered} with respect to $C$ if $C$ has at least one center from group $G_i$. Group $G_i$ is said to be {\em uncovered} otherwise.

\begin{lemma}
$(0, 0) \in L$.
\end{lemma}
\begin{proof}
Let $L' = \{(0, 0), (1, 0), ..., (k-1, 0)\}$. Then we have:
\begin{eqnarray*}
\Phi_{L'}(X) &=& \sum_{i=1}^{k-1} 2 \left(m_i r_i^2 + \frac{m_i}{4} (2 r_i)^2 + ... + \frac{m_i}{4^{k-1}} (2^{k-1}r_i)^2 \right)\\
 &=& \sum_{i=1}^{k-1} 2k \cdot m_i r_i^2 \\
 &=& 2k(k-1)m r^2.
\end{eqnarray*}
Let $L''$ be any set of locations that do not include $(0, 0)$, then $\Phi_{L''}(X) \geq 12k2^k m r^2$ (since the nearest location to $(0, 0)$ is $(1, 0)$).
So, $L$ necessarily includes the location $(0, 0)$.
\end{proof}

\begin{lemma}\label{lemma:axis}
For any $i$, if group $G_i$ is covered with respect to $C$, then $(i, 0) \in L$.
\end{lemma}
\begin{proof}
For the sake of contradiction, assume that $(i, 0) \notin L$. 
Let $(i, j) \in L$ be the location that is farthest from the $x$-axis among the locations of the form $(i, .) \in L$.
Consider the set of locations $L' = (L \setminus \{(i, j)\}) \cup \{(i, 0)\}$.
We will now show that $\Phi_{L'}(X) < \Phi_{L}(X)$.
WLOG let us assume that $j$ is positive.
The change in center location does not decrease the potential of $H_j, H_{j+1}, ...$, does not increase the potential of $H_{j-1}, H_{j-2}, ...$, and does not increase the potential of points on the $x$-axis. 
From Lemmas~\ref{lemma:3} and \ref{lemma:4}, we have that the increase in potential is at most $3kmr^2$. 
On the other hand, since the contribution of the points located at $(i, 0)$ to the total potential changes from $4kmr^2$ to $0$, the total decrease in potential is at least $4 k m r^2$. 
So, we have that the total potential decreases and hence $\Phi_{L'}(X) < \Phi_{L}(X)$. 
This contradicts the fact that $L$ denotes the location of the optimal centers.
\end{proof}

\begin{lemma}\label{lemma:distributed}
All groups are covered with respect to $C$.
\end{lemma}
\begin{proof}
For the sake of contradiction, assume that there is a group $G_i$ that is uncovered. 
This means that there is another group $G_j$ such that there are at least two locations from $G_j$ that is present in $L$.
Note that from the previous lemma $(j, 0) \in L$.
Let $(j, l) \in L$ for some $l > 0$.
We now consider the set of locations $L' = (L \setminus \{(j, l)\}) \cup \{(i, 0)\}$.
We will now show that $\Phi_{L'}(X) < \Phi_{L}(X)$.
Since $(j, 0) \in L$, the change in center location does not decrease the potential of $H_l, H_{l+1}, ...$, does not increase the potential of $H_{l-1}, H_{l-2}, ...$ and does not increase the potential of points on the $x$-axis. 
From Lemmas~\ref{lemma:3} and \ref{lemma:4}, we have that the increase in potential is at most $3kmr^2$. On the other hand, since the contribution of the points located at $(i, 0)$ to the total potential changes from $4kmr^2$ to $0$, the total decrease in potential is at least $4 k m r^2$. 
So, we have that the total potential decreases and hence $\Phi_{L'}(X) < \Phi_{L}(X)$. This contradicts the fact that $L$ denotes the location of the optimal centers.
\end{proof}

The following is a simple corollary of Lemmas~\ref{lemma:axis} and \ref{lemma:distributed}.

\begin{corollary}
Let $C$ denote the optimal set of centers for our k-means problem instance and let $L$ denote the location of these optimal centers. Then $L = \{(0, 0), (1, 0), ..., (k-1, 0)\}$.
\end{corollary}

\subsection{Potential of the optimal solution}
Let us denote the potential of the optimum solution by $\Phi^{*}$. 
Since optimum chooses its centers only from locations on the $x$-axis, we can compute $\Phi^{*}$ as follows:
\begin{eqnarray}\label{eqn:optimal-potential}
\Phi^{*} &=& \sum_{i=1}^{k-1} 2 \cdot (m_ir_i^2+\frac{m_i}{4}(2r_i)^2+\cdots+\frac{m_i}{4^{k-1}}(2^{k-1} r_i)^2) \nonumber\\
&=& \sum_{i=1}^{k-1} 2 k m_i r_i^2 \nonumber \\
&=& 2 k (k-1) m r^2
\end{eqnarray}

\section{Analysis of k-means++ for our instance}

We will first show that with very high probability, the first center chosen by the k-means++ seeding algorithm is located at the location $(0, 0)$.
This is simply due to the large number of points located at the location $(0, 0)$ and the fact that the first center is chosen uniformly at random from all the given points.

\begin{lemma}\label{lemma:event} 
Let $p$ be the location of the first center chosen by the k-means++ seeding algorithm. Then $\pr[p \neq (0, 0)] \leq 2^{-k}$.
\end{lemma}

\begin{proof}
For any $i \geq 1$ let $\omega(i) = 1 + 1/4 + ... + 1/4^{i-1} = (4/3) \cdot (1 - 1/4^i)$. 
Since the first center is chosen uniformly at random, we have:
{
\allowdisplaybreaks
\begin{eqnarray*}
\pr[p = (0, 0)] &=& \frac{M_0}{M_0 + M_1 + ... + M_{k-1}} \\
&=& \frac{M_0}{M_0 + \sum_{i=1}^{k-1} m_i \cdot (4k + 2\omega(k))} \\
&& \textrm{(since from (\ref{eqn:mass}), $M_i = m_i (4k + 2 \omega(k))$)} \\
&=& \frac{M_0}{M_0 + \sum_{i=1}^{k-1} \frac{m}{4^{i-1}} \cdot (4k + 2\omega(k))} \\
&=& \frac{M_0}{M_0 + m \cdot \omega(k-1) \cdot (4k + 2\omega(k))} \\
&=& \frac{(12 k) \cdot 2^{k}}{(12 k) \cdot 2^k + \omega(k-1) \cdot (4k + 2 \omega(k))} \\
&\geq& \frac{(12 k) \cdot 2^{k}}{(12 k) \cdot 2^k + (4/3) \cdot (4k + (8/3))}\\
&\geq& \frac{(12 k) \cdot 2^{k}}{(12 k) \cdot 2^k + 12k} \quad \textrm{(since $k \geq 1$)}\\
&\geq& 1 - 2^{-k}
\end{eqnarray*}
}
\end{proof}

Let us define the following event:
\begin{definition}\label{defn:event}
$\xi$ denotes the event that the location of the first chosen center is $(0, 0)$.
\end{definition}

Lemma~\ref{lemma:event} shows that $\xi$ happens with a very high probability.
We will do the remaining analysis conditioned on the event $\xi$. 
We will later use the above lemma to remove the conditioning. 
The advantage of using this event is that once the first center has the location $(0, 0)$, computing an upper-bound on the potential of any location becomes easy.
This is because we can compute potential with respect to the center at location $(0, 0)$. 
Computing such upper bounds will be crucial in our analysis.

Our analysis closely follows that of \cite{br12}. 
Let us analyze the situation after $(1 + t)$ iterations of the k-means++ seeding algorithm (given that the event $\xi$ happens).
Let $C_t$ denote the set of chosen centers.
Let $s \leq t$ denote the number of optimal clusters among $G_1, ..., G_{k-1}$ that are covered by $C_t$.
Let $X_c$ denote the points in these covered clusters and $X_u$ denote the points in the uncovered clusters.
Conditioned on $\xi$, the probability that the next center will be chosen from $X_u$ is $\frac{\Phi(X_u)}{\Phi(X_u) + \Phi(X_c)}$.
So, the probability of covering a previously uncovered cluster in iteration $(t+2)$ depends on the ratio $\frac{\Phi(X_u)}{\Phi(X_c)}$.
The smaller this ratio, the smaller is the chance of covering a new cluster.
We will show that this ratio is small for most iterations of the algorithm.
This means that even when the algorithm terminates, there are a number of uncovered clusters. 
This implies that the algorithm gives a solution that is worse compared to the optimal solution.
In order to upper-bound the ratio $\frac{\Phi(X_u)}{\Phi(X_c)}$, we will upper bound the value of $\Phi(X_u)$ and lower-bound the value of $\Phi(X_c)$.
We state these bounds formally in the next two lemmas.

\begin{lemma}\label{lemma:11}
$\Phi(X_c) \geq (2s - 1) \cdot \frac{k m r^2}{4}$.
\end{lemma}

\begin{proof}
For any covered cluster $G_i$ for $i>0$, we know that $G_i$ has points at levels $0, i-1, -i+1, i, -i, i+1, \ldots$.
For any such location $(i, j)$ (except location $(i, 0)$) such that $C_t$ does not have a center at this location, the contribution of the points at this location to $\Phi(X_c)$ is at least $\frac{m_i}{4^{|j|-1}} \cdot (2^{|j|-1} - \max(2^{|j|-2}, 1))^2 \cdot r_i^2 \geq mr^2/4 $. 
Furthermore, the contribution of points at location $(i, 0)$ in case $C_t$ does not contain a center from this location, is at least $mr^2$.
Therefore,
\begin{eqnarray*}
\Phi(X_c) &\geq& ((2k+1) \cdot s - t) \cdot \frac{mr^2}{4} \\
&\geq& (2s-1) \cdot \frac{k m r^2}{4} \quad \textrm{(since $t \leq k-1$)}
\end{eqnarray*}
\end{proof}

\begin{lemma}\label{lemma:12}
$\Phi(X_u) \leq (40 k) \cdot (k - s-1) m r^2 \Delta^2$.
\end{lemma}

\begin{proof}
Since the number of covered clusters among $G_1, ..., G_{k-1}$ is $s$, the number of uncovered clusters is given by $(k-s-1)$.
Let $G_i$ be any such uncovered cluster.
Since $\xi$ happens, there is a center at location $(0, 0)$.
Therefore, the contribution of $G_i$ to $\Phi(X_u)$ can be upper bounded by the quantity $\Phi_{\{(0, 0)\}}(G_i)$. 
This can be computed in the following manner:
\begin{eqnarray*}
&&\Phi_{\{(0, 0)\}}(G_i)\\
&=& \Phi_{\{(i, 0)\}}(G_i) + M_i \cdot \Delta^2  \cdot (r_1 + r_2 + ... + r_i)^2 \\
&=& \Phi_{\{(i, 0)\}}(G_i) + M_i \cdot \Delta^2 \cdot (2^i - 1)^2 \cdot r^2 \\
&=& \Phi_{\{(i, 0)\}}(G_i) + (4k + 2 \omega(k)) \frac{m}{4^{i-1}} \cdot \Delta^2 \cdot (2^i - 1)^2 \cdot r^2 \\
&\leq& \Phi_{\{(i, 0)\}}(G_i) + (4k + (8/3)) \cdot (4 m r^2 \Delta^2) \\
&=& 2 \cdot \sum_{j=1}^{k} \frac{m_i}{4^{j-1}} \cdot (2^{j-1}r_i)^2 + (4k + (8/3)) \cdot (4 m r^2 \Delta^2) \\
&=& 2 k m r^2 + (4k + (8/3)) \cdot (4 m r^2 \Delta^2) \\
&\leq& (40 k) m r^2 \Delta^2
\end{eqnarray*}
Hence, the total contribution from the uncovered clusters $\Phi(X_u)$ is upper bounded by 
$(40 k) (k-s-1)m r^2 \Delta^2$.
\end{proof}

We will also need a lower bound on $\Phi(X_u)$. This is given in the next lemma.

\begin{lemma}\label{lemma:13}
$\Phi(X_u) \geq 4k (k-s-1) m r^2 \Delta^2$.
\end{lemma}
\begin{proof}
Let $G_i$ be an uncovered cluster for some $i \geq 1$.
For any location $(i, j)$, the contribution of the points at this location to $\Phi(X_u)$ is at least $r_i^2 \Delta^2$ times the number of points at that location. So we have:
\begin{eqnarray*}
\Phi(X_u) &\geq& \sum_{\{i | G_i \textrm{ uncovered}\}} M_i \cdot r_i^2 \Delta^2 \\
&=& \sum_{\{i | G_i \textrm{ uncovered}\}} m_i (4k + 2 \omega(k)) \cdot r_i^2 \Delta^2 \\
&\geq& \sum_{\{i | G_i \textrm{ uncovered}\}} 4k \cdot m r^2 \Delta^2 \\
&\geq& 4k (k-s-1) \cdot m r^2 \Delta^2
\end{eqnarray*}
\end{proof}

Since most of our bounds have the term $k-1$, we define $\bar{k} = k-1$ and do the remaining analysis in terms of $\bar{k}$.
Note that all the bounds on $\Phi(X_u)$ and $\Phi(X_c)$ are dependent only on $s$ and not on $t$.
This allows us to define the following quantity that will be used in the remaining analysis. 
This is an upper bound on the ratio $\frac{\Phi_{u}(X)}{\Phi_{c}(X)}$ obtained from Lemmas~\ref{lemma:11} and \ref{lemma:12}.
\begin{equation}\label{eqn:z_s}
z_s \stackrel{def}{=} \frac{(\bar{k}-s)(80 \Delta^2)}{s-1/2} = \frac{(k-s-1)(80 \Delta^2)}{s-1/2}
\end{equation}

We now get a bound on the number of clusters among $G_1, ..., G_{\bar{k}}$ that are needed to be covered to achieve an approximation factor of $\alpha$ for a fixed $\alpha$.
For any such fixed approximation factor $\alpha$, we define the following quantities that will be used in the analysis.
\begin{equation}
u \stackrel{def}{=} \frac{\alpha}{2 \Delta^2} \quad \textrm{and} \quad 
s^* \stackrel{def}{=} \lceil \bar{k} \cdot (1 - u) \rceil
\end{equation}

\begin{lemma}
Any $\alpha$-approximate clustering covers $G_0$ and at least $s^*$ clusters among $G_1, ..., G_{\bar{k}}$.
\end{lemma}
\begin{proof}
The optimal potential is given by $\Phi^{*} = 2k \bar{k}m r^2$ (by (\ref{eqn:optimal-potential})).
Consider any $\alpha$-approximate clustering. 
Suppose this clustering covers $s$ clusters among $G_1, ..., G_{\bar{k}}$. 
Let the covered and uncovered clusters be denoted by $X_c$ and $X_u$ respectively.
Then we have:
\[
\alpha = \frac{\Phi(X)}{\Phi^{*}} \geq \frac{\Phi(X_u)}{\Phi^{*}} \geq \frac{4k (\bar{k}-s)m r^2 \Delta^2}{2k \bar{k} m r^2} \geq \frac{2 (\bar{k}-s) \Delta^2}{\bar{k}}
\]
The second inequality above is using Lemma~\ref{lemma:13}.
This means that the number of covered clusters among $G_1, ..., G_{k-1}$ should satisfy
\[
s \geq \left\lceil \bar{k} \cdot \left( 1 - \frac{\alpha}{2 \Delta^2} \right) \right\rceil = s^*.
\]
\end{proof}

\begin{figure}
\centering
\includegraphics[scale=0.3]{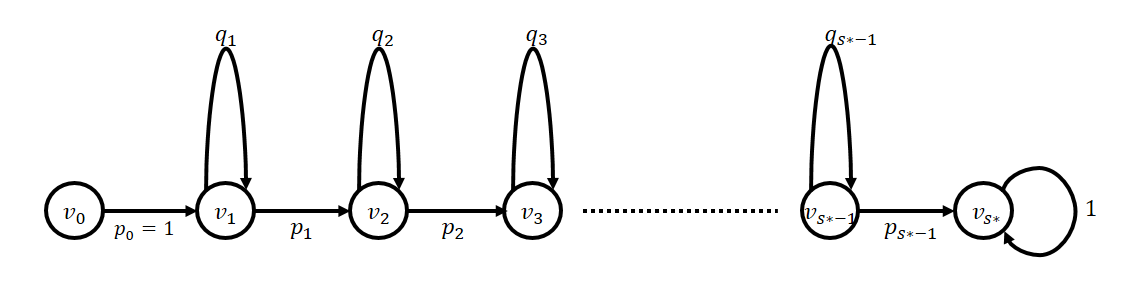}
\caption{Markov chain used for analyzing the algorithm.}
\label{fig:2}
\end{figure}

We analyze the behavior of the k-means++ seeding algorithm with respect to the number of covered optimal clusters using a Markov chain (see Figure~\ref{fig:2}). 
This Markov chain is almost the same as the Markov chain used to analyze the bad instance by Brunsch and R\"{o}glin~\cite{br12}.
In fact, the remaining analysis will mostly mimic that analysis in \cite{br12}. 
The next lemma formally relates the probability that the algorithm achieves an $\alpha$ approximation to the Markov chain reaching its end state.
We analyze this Markov chain in the next subsection.

\begin{lemma}
Let $p_0 = 1$ and for $s = 1, 2, ..., s^*$, let 
$
p_s = \frac{1}{1 + \frac{1}{z_s}}
$
We consider the linear Markov chain with states $v_0, v_1, ..., v_{s^*}$ with starting state $v_0$ (see Figure~\ref{fig:2}).
Edges $(v_s, v_{s+1})$ have transition probabilities $p_s$ and the self-loops $(v_s, v_s)$ have transition probabilities $q_s = (1 - p_s)$.
Then the probability that the k-means++ seeding algorithm gives an $\alpha$-approximate solution is upper bounded by the probability that the state $v_{s^*}$ is reached by the Markov chain within $\bar{k}$ steps.
\end{lemma}
\begin{proof}
The proof is trivial from the observation that the probability that a previously uncovered cluster will be covered in iteration $i >2$ is given by $\frac{\Phi(X_u)}{\Phi(X_c) + \Phi(X_u)} \leq \frac{1}{1 + \frac{1}{z_s}} = p_s$.
\end{proof}

\subsection{Definitions and inequalities}
A number of quantities will be used for the analysis of the Markov chain. 
The reader is advised to refer to this subsection when reading the next subsection dealing with the analysis of the Markov chain.
The following quantities written as a function of $\bar{k}$ will be used in the analysis:
\begin{eqnarray}
\alpha(\bar{k}) &=& \delta \cdot \log{\bar{k}} \label{eqn:E1} \\
 \eps(\bar{k}) &=& \frac{1}{120} \cdot \frac{\log{\alpha(\bar{k})}}{\alpha(\bar{k})} \label{eqn:E2} \\
 \Delta(\bar{k}) &=& \left\lceil \sqrt{\alpha(\bar{k})} \cdot \exp \left(80 \cdot \alpha(\bar{k}) \cdot \frac{1 + \eps(\bar{k})}{4} \right)\right\rceil  \label{eqn:E3} \\
 u(\bar{k}) &=& \frac{\alpha(\bar{k})}{2 \Delta^2(\bar{k})}  \label{eqn:E4} \\
 s^*(\bar{k}) &=& \lceil \bar{k} \cdot (1 - u(\bar{k}))\rceil  \label{eqn:E5}\\
 z_s(\bar{k}) &=& \frac{(\bar{k}-s) \cdot (80 \Delta^2)}{s - 1/2} \label{eqn:E6}\\
 p_s(\bar{k}) &=& \frac{1}{1 + \frac{1}{z_s(\bar{k})}} \label{eqn:E7}
\end{eqnarray}

We will also use the following inequalities. Here, whenever we say that $f(\bar{k}) \leq g(\bar{k})$ for two functions $f$ and $g$, we actually mean to say that $f(\bar{k}) \leq g(\bar{k})$ for all sufficiently large $\bar{k}$.
\begin{eqnarray}
&& \frac{1}{k} \leq u(\bar{k}) < \frac{1}{2} \label{eqn:I1}\\
&& (1 + 40 \alpha(\bar{k}))^{\Delta(\bar{k})} \geq \frac{1}{u^2(\bar{k})} \label{eqn:I2}\\
&& \frac{1}{\bar{k}} \leq \frac{\eps(\bar{k})}{9} \label{eqn:I3}\\
&& \frac{1}{80 \Delta^2(\bar{k})} \leq \frac{\eps(\bar{k})}{3} \cdot u(\bar{k}) \label{eqn:I4}\\
&& u(\bar{k}) + \frac{\eps(\bar{k})}{3} \cdot \left( 1 + \frac{\eps(\bar{k})}{3}\right)\cdot u^2(\bar{k}) \leq \left( \frac{\eps(\bar{k})}{3}\right)^2 \label{eqn:I5}
\end{eqnarray}

Except for inequality (\ref{eqn:I2}), all the inequalities are the same as in \cite{br12}. 
We refer the reader to \cite{br12} for the correctness of these inequalities.
As for (\ref{eqn:I2}), note that $(1 + 40 \alpha(\bar{k}))^{\Delta(\bar{k})} \geq 2^{\Delta(\bar{k})} = 2^{\Omega(\sqrt{\alpha(\bar{k})} \cdot e^{\alpha(\bar{k})/4})}$ and $1/u^2(\bar{k}) = O(e^{2\alpha(\bar{k})})$.
So, for sufficiently large values of $\bar{k}$, the inequality is true.

For the remaining analysis, we will assume that the value of $\bar{k}$ is fixed such that inequalities (\ref{eqn:I1}), (\ref{eqn:I2}), (\ref{eqn:I3}), (\ref{eqn:I4}), and (\ref{eqn:I5}) are true. 
Given this, we will avoid using the functional notation and simply use the name of the quantities. For example, we will use $u$ instead of $u(\bar{k})$ and $\eps$ instead of $\eps(\bar{k})$ etc.

\subsection{Analysis of Markov chain}

We now analyze the Markov chain and upper bound the probability of this Markov chain reaching the state $v_{s^*}$ within $\bar{k}$ steps. 
To be able to do so, we define random variables $X_0, X_1, ..., X_{s^*-1}$, where the $X_s$ denotes the number of steps to move from state $v_s$ to state $v_{s+1}$.
We consider the random variable $X = \sum_{s=0}^{s^* - 1} X_s$.
We would like to show that the expected value of $X$ is much larger than $\bar{k}$ and then use the Hoeffding inequality to bound the probability. 
To do this using the well known Hoeffding bound, we need to have a bound on the value of each of the random variables. 
So, we define related random variables $Y_0, Y_1, ..., Y_{s^*-1}$, where $Y_s = \min(X_s, \Delta)$. 
We will analyze the random variable $Y = \sum_{s=0}^{s^*-1} Y_s \leq X$.
We will use the following lemma from \cite{br12}.

\begin{lemma}[Claim 5 from \cite{br12}]
The expected value of $X_s$ is $1/p_s$ and the expected value of $Y_s$ is $(1-q_s^{\Delta})/p_s$.
\end{lemma}

The next lemma relates the expected values of $X_s$ and $Y_s$.

\begin{lemma}[Similar to Lemma 6 in \cite{br12}]
$\frac{\ex[Y_s]}{\ex[X_s]} \geq 1 - u^2$.
\end{lemma}
\begin{proof}
First we get a lower bound on $z_s$ in the following manner:
\begin{eqnarray*}
z_s &=& \frac{(\bar{k}-s) (80 \Delta^2)}{s - 1/2} \\
%
%
&\geq&  \frac{u \cdot (80 \Delta^2)}{1 - u - \frac{1}{2\bar{k}}} \quad \textrm{(since $s \leq s^*-1 \leq \bar{k}(1-u)$)} \\
&=& \frac{40 \alpha}{1 - u - \frac{1}{2\bar{k}}} \quad \textrm{(using (\ref{eqn:E4}))} \\
&\geq& 40 \alpha \quad \textrm{(using (\ref{eqn:I1}))}
\end{eqnarray*}
Also, from the previous lemma we have:
\begin{eqnarray*}
\frac{\ex[Y_s]}{\ex[X_s]} &=& 1 - q_s^{\Delta}= 1 - (1 - p_s)^{\Delta} = 1 - \left( \frac{1}{1 + z_s} \right)^{\Delta}\\
			  &\geq& 1 - \left( \frac{1}{1 + 40 \alpha} \right)^{\Delta} \geq 1 - u^2.
\end{eqnarray*}
The last inequality used (\ref{eqn:I2}).
\end{proof}

Next, we we get a lower bound on $\ex[X]$.

\begin{lemma}[Similar to Lemma 7 in \cite{br12}]
$\frac{\ex[X]}{\bar{k}} \geq 1 + \frac{\eps}{3} \cdot \left( 1 + \frac{\eps}{3}\right)\cdot u$.
\end{lemma}
\begin{proof}
We can lower-bound $\ex[X]$ in the following manner:
\begin{eqnarray*}
&&\ex[X]=\sum_{s=0}^{s^*-1} \frac{1}{p_s}\\
&=& 1 + \sum_{s=1}^{s^*-1} \left( 1 + \frac{s-1/2}{(\bar{k}-s) (80 \Delta^2)}\right) \\
&=& s^* + \sum_{i=\bar{k}-s^*+1}^{k-1} \frac{\bar{k}-1/2-i}{i \cdot (80 \Delta^2)} \\
&\geq& s^* - \frac{s^*-1}{80 \Delta^2} + \frac{\bar{k}-1}{80 \Delta^2} \cdot \sum_{i=\bar{k}-s^*+1}^{\bar{k}-1} \frac{1}{i} \\
&\geq& s^* \cdot \left(1 - \frac{1}{80 \Delta^2}\right) + \frac{\bar{k}-1}{80 \Delta^2} \cdot \log \left( \frac{\bar{k}}{\bar{k}-s^*+1}\right) \\
\end{eqnarray*}
Since $s^* \geq \bar{k}(1-u)$, we can write,
\begin{eqnarray*}
&&\ex[X]\\
&\geq& \bar{k} (1-u) \left(1 - \frac{1}{80 \Delta^2}\right) + \frac{\bar{k}-1}{80 \Delta^2} \cdot \log \left( \frac{\bar{k}}{\bar{k}-\bar{k}(1-u)+1}\right)\\
&\geq& \bar{k} \left(1 - u - \frac{1}{80 \Delta^2} + \frac{\bar{k}-1}{\bar{k}} \cdot \frac{1}{80 \Delta^2} \cdot \log \left( \frac{1}{u + \frac{1}{\bar{k}}}\right)\right)
\end{eqnarray*}

Using this, we have:
\begin{eqnarray*}
&&\frac{\ex[X]}{\bar{k}}\\
&\geq& \left(1 - u - \frac{\eps u}{3} + \frac{\bar{k}-1}{\bar{k}} \cdot \frac{1}{80  \Delta^2} \cdot \log \left( \frac{1}{u + \frac{1}{\bar{k}}}\right)\right) \quad \textrm{(using (\ref{eqn:I4}))}\\
&\geq& \left(1 - u \left(1+\frac{\eps}{3}\right)  + \frac{\bar{k}-1}{\bar{k}} \cdot \frac{1}{80 \Delta^2} \cdot \log \left( \frac{1}{2u}\right)\right) \quad \textrm{(using (\ref{eqn:I1}))}\\
&\geq& \left(1 - u \left(1+\frac{\eps}{3}\right)  + \left(1 -\frac{\eps}{9}\right)\cdot \frac{1}{80 \Delta^2} \cdot \log \left( \frac{1}{2u}\right)\right) \quad \textrm{(using (\ref{eqn:I3}))}\\
&=& \left(1 - u \left(1+\frac{\eps}{3}\right)  + \left(1 - \frac{\eps}{9}\right)\cdot \frac{1}{80 \Delta^2} \cdot \log \left( \frac{\Delta^2}{\alpha}\right)\right) \quad \textrm{(using (\ref{eqn:E4}))}\\
&\geq& \left(1 - u \left(1+\frac{\eps}{3}\right)  + \left(1 - \frac{\eps}{9}\right) \frac{1}{\Delta^2} \cdot \alpha \cdot \frac{1+\eps}{2}\right) \quad \textrm{(using (\ref{eqn:E3}))}\\
&=& \left(1 - u \left(1+\frac{\eps}{3}\right)  + \left(1 - \frac{\eps}{9}\right) (1+\eps) u\right) \quad \textrm{(using (\ref{eqn:E4}))}\\
&=& 1 + \frac{\eps}{3} \left( 1 + \frac{\eps}{3}\right)  u
\end{eqnarray*}
\end{proof}

Using the previous two lemmas, we can now obtain a lower bound on $\ex[Y]$.

\begin{lemma}[Same as Corollary 8 in \cite{br12}]\label{lemma:19}
$\frac{\ex[Y]}{\bar{k}} \geq 1 + \frac{\eps}{3} \cdot u$.
\end{lemma}
\begin{proof}
Using the last two lemmas, we have
\begin{eqnarray*}
\frac{\ex[Y]}{\bar{k}} &\geq& (1 - u^2) \cdot \frac{\ex[X]}{\bar{k}} \\
&\geq& (1 - u^2) \cdot \left( 1 + \frac{\eps}{3} \left( 1 + \frac{\eps}{3}\right) u \right) \\
&=& 1 + u \cdot \left( \frac{\eps}{3} \left(1 + \frac{\eps}{3}\right) - \left(u + \frac{\eps}{3} \left( 1 + \frac{\eps}{3}\right) u^2\right)\right)\\
&\geq& 1 + u \cdot \left( \frac{\eps}{3}  \left(1 + \frac{\eps}{3}\right) - \left(\frac{\eps}{3}\right)^2\right) \quad \textrm{(using (\ref{eqn:I5}))}\\
&=& 1 + \frac{\eps}{3} \cdot u
\end{eqnarray*}
\end{proof}

We can finally bound the probability that the Markov chain reaches the state $v_{s^*}$.

\begin{lemma}[Similar to Lemma~9 in \cite{br12}]\label{lemma:20}
The probability that the state $v_{s^*}$ is reached within $\bar{k}$ steps is bounded by $\exp(-\bar{k}^{1 - 120\delta - o(1)})$.
\end{lemma}
\begin{proof}
The bound on the probability is obtained through the following calculations:
\begin{eqnarray*}
\pr[X \leq \bar{k}] &\leq& \pr[Y \leq \bar{k}] \quad \textrm{(since $Y \leq X$)} \\
&\leq& \pr\left[\ex[Y] - Y \geq \frac{\eps}{3} \cdot u \cdot \bar{k}\right] \quad \textrm{(by Lemma~\ref{lemma:19})} \\
&\leq& \exp\left( - \frac{2 \cdot (\frac{\eps}{3} \cdot u \cdot \bar{k})^2}{s^* \Delta^2} \right) \quad \textrm{(by Hoeffding bound)}\\
&\leq& \exp\left( - \frac{2 \eps^2u^2\bar{k}^2}{9 k \Delta^2} \right) \\
&=& \exp\left( - \bar{k} \cdot \frac{2 \eps^2u^2}{9 \Delta^2} \right)
\end{eqnarray*}
We will now get a bound on $\frac{2 \eps^2 u^2}{9 \Delta^2}$.
\begin{eqnarray*}
\frac{2 \eps^2 u^2}{9 \Delta^2} &=& \frac{\eps^2 \alpha^2}{18 \Delta^6} \\
&=& \frac{\eps^2 \alpha^2}{18 \cdot \alpha^3 \cdot \exp\left(80 \cdot 6 \cdot \alpha \cdot \frac{1 + \eps}{4}\right)} \quad \textrm{(using (\ref{eqn:E3}))}\\
&=& \frac{\eps^2 \cdot \alpha^{-2} \cdot e^{-120\alpha}}{18} \\
&=& \bar{k}^{-o(1)} \cdot \bar{k}^{-o(1)} \cdot \bar{k}^{-120 \delta}
\end{eqnarray*}
\end{proof}

Now we can put everything together and prove our main theorem.

\begin{proof}[(Proof of main theorem)]
Given that the event $\xi$ occurs, the probability that the k-means++ seeding algorithm gives an approximation factor of at most $(\delta \cdot \log{(k-1)})$ is upper bounded by the probability that the Markov chain reaches the state $v_{s^*}$ in at most $(k-1)$ steps. 
This is bounded by $\exp(-(k-1)^{1 - o(1) - 120 \delta})$ from Lemma~\ref{lemma:20}.
Also, from Lemma~\ref{lemma:event}, we know that $\pr[\neg \xi] \leq 2^{-k}$.
Combining these, we get that the probability that the algorithm gives an approximation factor of $(\delta \cdot \log{k})$ is at most $2^{-k} + \exp(-(k-1)^{1 - o(1) - 120 \delta})$
\end{proof}

\section{Acknowledgements}
Ragesh Jaiswal would like to thank Prachi Jain, Saumya Yadav, Nitin Garg, and Abhishek Gupta for helpful discussions.

\bibliographystyle{abbrv}
\bibliography{paper}

\begin{thebibliography}{10}

\bibitem{ab10}
M.~R. Ackermann and J.~Bl\"{o}mer.
\newblock Bregman clustering for separable instances.
\newblock In {\em Proceedings of the 12th Scandinavian conference on Algorithm
  Theory}, SWAT'10, pages 212--223, Berlin, Heidelberg, 2010. Springer-Verlag.

\bibitem{ajp13}
M.~Agarwal, R.~Jaiswal, and A.~Pal.
\newblock k-means++ under approximation stability.
\newblock In T.-H. Chan, L.~Lau, and L.~Trevisan, editors, {\em Theory and
  Applications of Models of Computation}, volume 7876 of {\em Lecture Notes in
  Computer Science}, pages 84--95. Springer Berlin Heidelberg, 2013.

\bibitem{AggarwalDK09}
A.~Aggarwal, A.~Deshpande, and R.~Kannan.
\newblock Adaptive sampling for k-means clustering.
\newblock In I.~Dinur, K.~Jansen, J.~Naor, and J.~Rolim, editors, {\em
  Approximation, Randomization, and Combinatorial Optimization. Algorithms and
  Techniques}, volume 5687 of {\em Lecture Notes in Computer Science}, pages
  15--28. Springer Berlin Heidelberg, 2009.

\bibitem{AJMonteleoni09}
N.~Ailon, R.~Jaiswal, and C.~Monteleoni.
\newblock Streaming k-means approximation.
\newblock In {\em NIPS}, pages 10--18. 2009.

\bibitem{ArthurV06}
D.~Arthur and S.~Vassilvitskii.
\newblock How slow is the k-means method?
\newblock In {\em Proceedings of the twenty-second annual symposium on
  Computational geometry}, SCG '06, pages 144--153, New York, NY, USA, 2006.
  ACM.

\bibitem{ArthurV06b}
D.~Arthur and S.~Vassilvitskii.
\newblock Worst-case and smoothed analysis of the {ICP} algorithm, with an
  application to the k-means method.
\newblock In {\em Proceedings of the 47th Annual IEEE Symposium on Foundations
  of Computer Science}, FOCS '06, pages 153--164, Washington, DC, USA, 2006.
  IEEE Computer Society.

\bibitem{ArthurV07}
D.~Arthur and S.~Vassilvitskii.
\newblock k-means++: the advantages of careful seeding.
\newblock In {\em Proceedings of the eighteenth annual ACM-SIAM symposium on
  Discrete algorithms}, SODA '07, pages 1027--1035, Philadelphia, PA, USA,
  2007. Society for Industrial and Applied Mathematics.

\bibitem{b12}
B.~Bahmani, B.~Moseley, A.~Vattani, R.~Kumar, and S.~Vassilvitskii.
\newblock Scalable k-means++.
\newblock {\em Proc. VLDB Endow.}, 5(7):622--633, Mar. 2012.

\bibitem{br12}
T.~Brunsch and H.~R\"{o}glin.
\newblock A bad instance for k-means++.
\newblock {\em Theoretical Computer Science}, 2012.

\bibitem{d07}
S.~Dasgupta.
\newblock The hardness of k-means clustering.
\newblock Technical report, University of California San Diego.

\bibitem{jg12}
R.~Jaiswal and N.~Garg.
\newblock Analysis of k-means++ for separable data.
\newblock In A.~Gupta, K.~Jansen, J.~Rolim, and R.~Servedio, editors, {\em
  Approximation, Randomization, and Combinatorial Optimization. Algorithms and
  Techniques}, volume 7408 of {\em Lecture Notes in Computer Science}, pages
  591--602. Springer Berlin Heidelberg, 2012.

\bibitem{jks12}
R.~Jaiswal, A.~Kumar, and S.~Sen.
\newblock A simple ${D}^2$-sampling based {PTAS} for k-means and other
  clustering problems.
\newblock {\em Algorithmica}, 2013.

\bibitem{mnv12}
M.~Mahajan, P.~Nimbhorkar, and K.~Varadarajan.
\newblock The planar k-means problem is {NP}-hard.
\newblock {\em Theoretical Computer Science}, 442(0):13 -- 21, 2012.
\newblock Special Issue on the Workshop on Algorithms and Computation (WALCOM
  2009).

\bibitem{v09}
A.~Vattani.
\newblock The planar k-means problem is {NP}-hard.
\newblock {\em Manuscript}, 2009.

\end{thebibliography}
\end{document}